\documentclass{article}
\usepackage{cite,makeidx,array}
\usepackage[table]{xcolor}
\usepackage{tikz}
\usetikzlibrary{arrows,backgrounds,snakes}
\usepackage{graphicx}
\usepackage{multirow,multicol}
\usepackage{subfigure}
\usepackage[bottom]{footmisc}
\usepackage{amsmath,amssymb,amsthm}
\usepackage[linesnumbered,ruled,vlined]{algorithm2e}
\usepackage{algpseudocode,url}
\algnewcommand{\algorithmicgoto}{\textbf{go to}}%
\algnewcommand{\Goto}[1]{\algorithmicgoto~\ref{#1}}%
\DeclareMathOperator*{\argmin}{\arg\!\min}

\newtheorem{theorem}{Theorem}
\newtheorem{lemma}{Lemma}
\newtheorem{definition}{Definition}

\usepackage[plainpages=false,pdfpagelabels]{hyperref}

\begin{document}

\title{A Novel Approach to the Common Due-Date Problem on Single and Parallel Machines}
\author{Abhishek Awasthi$^*$, J\"org L\"assig$^*$ and Oliver Kramer$^\dag$\\ \\
$^*$Department of Computer Science\\
University of Applied Sciences Zittau/G\"orlitz \\ G\"orlitz, Germany\\
\{abhishek.awasthi, joerg.laessig\}@hszg.de
\and
$^\dag$Department of Computing Science\\
Carl von Ossietzky University of Oldenburg\\
Oldenburg, Germany \\
oliver.kramer@uni-oldenburg.de
}
\maketitle

\begin{abstract}
This chapter presents a novel idea for the general case of the Common Due-Date (CDD) scheduling problem. The problem is about scheduling a certain number of jobs on a single or parallel machines where all the jobs possess different processing times but a common due-date. The objective of the problem is to minimize the total penalty incurred due to earliness or tardiness of the job completions. This work presents exact polynomial algorithms for optimizing a given job sequence for single and identical parallel machines with the run-time complexities of $O(n \log n)$ for both cases, where $n$ is the number of jobs. Besides, we show that our approach for the parallel machine case is also suitable for non-identical parallel machines. We prove the optimality for the single machine case and the runtime complexities of both. Henceforth, we extend our approach to one particular dynamic case of the CDD and conclude the chapter with our results for the benchmark instances provided in the OR-library.
\end{abstract}

\section{Introduction}
The Common Due-Date scheduling problem involves sequencing and scheduling of jobs over machine(s) against a common due-date. Each job possesses a processing time and different penalties per unit time in case the job is completed before or later than the due-date. The objective of the problem is to schedule the jobs so as to minimize the total penalty due to earliness or tardiness of all the jobs. In practice, a common due date problem occurs in almost any manufacturing industry. Earliness of the produced goods is not desired because it requires the maintenance of some stocks leading to some expenses to the industry for storage cost, tied-up capital with no cash flow \emph{etc.}. On the other hand, a tardy job leads to customer dissatisfaction.

When scheduling on a single machine against a common due date, one job at most can be completed exactly at the due date. Hence, some of the jobs will complete earlier than the common due-date, while other jobs will finish later. Generally speaking, there are two classes of the common due-date problem which have proven to be NP-hard, namely: 
\begin{itemize}
\item Restrictive CDD problem
\item Non-restrictive CDD problem.
\end{itemize}

A CDD problem is said to be \textit{restrictive} when the optimal value of the objective function depends on the due-date of the problem instance. In other words, changing the due date of the problem changes the optimal solution as well. However, in the \textit{non-restrictive} case a change in the value of the due-date for the problem instance does not affect the solution value. It can be easily proved that in the restrictive case, the sum of the processing times of all the jobs is strictly greater than the due date and in the non-restrictive case the sum of the processing times is less than or equal to the common due-date.

In this chapter, we study the restrictive case of the problem. However, our approach can be applied to the non-restrictive case on the same lines. We consider the scenario where all the jobs are processed on one or more machines without pre-emption and each job possesses different earliness/tardiness penalties. We also discuss a particular dynamic case of the CDD on a single machine and prove that our approach is optimal with respect to the solution value.

\section{Related Work}
The Common due-date problem has been studied extensively during the last $30$ years with several variants and special cases~\cite{seidmann,kanet}. In $1981$, Kanet presented an O($n \log n $) algorithm for minimizing the total absolute deviation of the completion of jobs from the due date for the single machine, $n$ being the number of jobs~\cite{kanet}. Panwalkar~\emph{et al.} considered the problem of common due-date assignment to minimize the total penalty for one machine~\cite{panwalkar}. The objective of the problem was to determine the optimum value for the due-date and the optimal job sequence to minimize the penalty function, where the penalty function also depends on the due-date along with earliness and tardiness. An algorithm of O($n \log n$) complexity was presented but the special problem considered by them consisted of symmetric costs for all the jobs~\cite{seidmann,panwalkar}.

Cheng again considered the same problem with slight variations and presented a linear programming formulation~\cite{cheng1}. In $1991$ Cheng and Kahlbacher and Hall~\emph{et al.} studied the CDD problem extensively, presenting some useful properties for the general case~\cite{cheng,hall}. A pseudo polynomial algorithm of O($n^2d$) (where, $d$ is the common due-date) complexity was presented by Hoogeveen and Van de Velde for the restrictive case with one machine when the earliness and tardiness penalty weights are symmetric for all the jobs~\cite{hoogeveen}. In $1991$ Hall~\emph{et al.} studied the unweighted earliness and tardiness problem and presented a dynamic programming algorithm~\cite{hall}. Besides these earlier works, there has been some research on heuristic algorithms for the general common due date problem with asymmetric penalty costs. James presented a tabu search algorithm for the general case of the problem in $1997$~\cite{james}.

More recently in $2003$, Feldmann and Biskup approached the problem using metaheuristic algorithms namely simulated annealing (SA) and threshold accepting and presented the results for benchmark instances up to $1000$ jobs on a single machine~\cite{biskup,feldmann}. Another variant of the problem was studied by Toksari and G\"uner in $2009$, where they considered the common due date problem on parallel machines under the effects of time dependence and deterioration~\cite{toksari}. Ronconi and Kawamura proposed a branch and bound algorithm in $2010$ for the general case of the CDD and gave optimal results for small benchmark instances~\cite{ronconi}. In $2012$, Rebai~\emph{et al.} proposed metaheuristic and exact approaches for the common due date problem to schedule preventive maintenance tasks~\cite{rebai}. 

In $2013$, Banisadr~\emph{et al.} studied the single-machine scheduling problem for the case that each job is considered to have linear earliness and quadratic tardiness penalties with no machine idle time. They proposed a hybrid approach for the problem based upon evolutionary algorithm concepts~\cite{Banisadr}. Yang~\emph{et al.} investigated the single-machine multiple common due dates assignment and scheduling problems in which the processing time of any job depends on its position in a job sequence and its resource allocation. They proposed a polynomial algorithm to minimize the total penalty function containing earliness, tardiness, due date, and resource consumption costs~\cite{Yang}.

This chapter is an extension of a research paper presented by the same authors in~\cite{cdd}. We extend our approach for a dynamic case of the problem and for non-identical parallel machines. Useful examples for both the single and parallel machines case are presented.

\section{Problem Formulation}\label{probform}
In this Section we give the mathematical notation of the common due date problem based on~\cite{biskup}. We also define some new parameters which are necessary for our considerations later on.

\noindent
Let, \\
$n$  =  total number of jobs \\
$m$ = total number of machines \\
$n_{j}$ = number of jobs processed by machine $j$ $( j = 1,2,\dots,m)$  \\
$M_{j}$ = time at which machine $j$ finished its latest job \\
$W_{j}^{k}$ = $k^{th}$ job processed by machine $j$ \\
$P_{i}$ = processing time of job $i$ $( i = 1,2,\dots,n)$ \\
$C_{i}$ = completion time of job $i$ $( i = 1,2,\dots,n)$ \\
$D$ = the common due date \\
$\alpha_{i}$ = the penalty cost per unit time for job $i$ for being early \\ 
$\beta_{i}$ = the penalty cost per unit time for job $i$ for being tardy \\ 
$E_{i}$ = earliness of job $i$, $E_i = \max\{0,D-C_{i}\}$ ($i = 1,2,\dots,n$) \\
$T_{i}$ = tardiness of job $i$, $T_{i}=\max\{0,C_{i}-D\}$ ($i = 1,2,\dots,n$)\;.

The cost corresponding to job $i$ is then expressed as $\alpha_{i} \cdot E_{i}+\beta_{i} \cdot T_{i}$. If job $i$ is completed at the due date then both $E_{i}$ and $T_{i}$ are equal to zero and the cost assigned to it is zero. When job $i$ does not complete at the due date, either $E_{i}$ or $T_{i}$ is non-zero and there is a strictly positive cost incurred. The objective function of the problem can now be defined as
\begin{equation}\label{ob}
\min \sum\limits_{i=1}^{n} (\alpha_{i} \cdot E_{i}+\beta_{i} \cdot T_{i}) \;.
\end{equation}

According to the 3-field problem classification introduced by Graham~\emph{et al.}~\cite{Graham}, the common due-date scheduling problem on a single machine can be expressed as $1|P_i|\sum_{\substack{i=1}}^{n} (\alpha_{i}E_{i}+\beta_{i}T_{i})$. This three field notation implies that the jobs with different processing times are scheduled on a single machine to minimize the total earliness and tardiness penalty.

\section{The Exact Algorithm for a Single Machine}
We now present the ideas and the algorithm for solving the single machine case for a given job sequence. From here onwards we assume that there are $n$ jobs to be processed by a machine and all the parameters stated at the beginning of Section~\ref{probform} represent the same meaning. The intuition for our approach comes from a property presented and proved by Cheng and Kahlbacher for the CDD problem~\cite{cheng}. They proved that the optimal solution for a problem instance with general penalties has no idle time between any two consecutive jobs or in other words, when the schedule is compact. This property implies that at no point of time the machine processing the jobs is left idle till the processing of all the jobs is completed. In our approach we first initialize the completion times of all the jobs without any idle times and then shift all the jobs with the same amount of time.

Let $J$ be the input job sequence where $J_i$ is the $i$th job in the sequence $J$. Note that without loss of any generality we can assume $J_i=i$, since we can rank the jobs for any sequence as per their order of processing. The algorithm takes the job sequence $J$ as the input and returns the optimal value for Equation~\eqref{ob}. There are three requirements for the optimal solution: allotment of jobs to specific machines, the order of processing of jobs in every machine and the completion times for all the jobs.

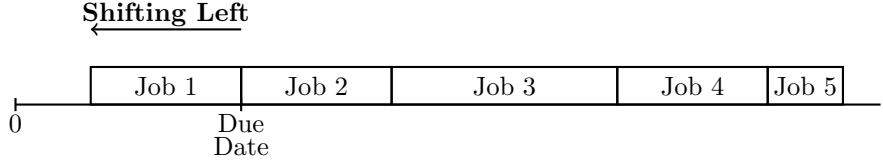
\begin{figure}[bht]
\centering
\begin{tikzpicture}[thick]
\draw (-6,1) -- (5.5,1);
\draw[thick,snake=ticks,segment length=3cm] (-6,1) -- (-3,1);
\draw[<-] (-5,2) -- (-3,2);
\node at (-3 , 0.75) {Due};
\node at (-3 , 0.45) {Date};
\node at (-6 , 0.75) {0};
\filldraw[fill=none,thick] (-5,1) rectangle (-3,1.5) node[midway] {Job 1};
\filldraw[fill=none,thick] (-3,1) rectangle (-1,1.5) node[midway] {Job 2};
\filldraw[fill=none,thick] (-1,1) rectangle (2,1.5) node[midway] {Job 3};
\filldraw[fill=none,thick] (2,1) rectangle (4,1.5) node[midway] {Job 4};
\filldraw[fill=none,thick] (4,1) rectangle (5,1.5) node[midway] {Job 5};
\draw	(-4,2.2) node{\textbf{Shifting Left}};
\end{tikzpicture}
\caption{Left shift (decrease in completion times) of all the jobs towards decreasing total tardiness for a sequence with $5$ jobs. Each reduction is done by the minimum of the processing time of the job which is starting at the due date and the maximum possible left shift for the first job.}
\label{fig1}
\end{figure}

Using the property of compactness proved by Cheng and Kahlbacher~\cite{cheng}, our algorithm assigns the completion times to all the jobs such that the first job is finished at $\max\{P_{1},D\}$ and the rest of the jobs follow without any idle time in order to obtain an initial solution which is then improved incrementally. It is quite apparent that a better solution for this sequence can be found only by reducing the completion times of all the jobs, {\em i.e.\ \/}shifting all the jobs towards decreasing total tardiness penalty as shown in Figure~\ref{fig1} with five jobs. Shifting all the jobs to the right will only increase the total tardiness. 

Hence, we first assign the jobs in $J$ to the machine such that none of the jobs are early and there is no idle time between the processing of any two consecutive jobs, as stated in Equation~\eqref{initial}.
\begin{equation}\label{initial}
C_{i} = 
\begin{cases}
\max\{P_{1},D\} & \mbox{if } i=1 \\ 
C_{i-1}+P_{i} & \mbox{if } 2 \leq i \leq n \;.
\end{cases}
\end{equation}

Before stating the exact algorithm for a given sequence for the single machine case, algorithm we first introduce some new parameters, definitions and theorems which are useful for the description of the algorithm. We first define $DT_{i} = C_{i}-D$, $i=1,2,\dots,n$ and $ES=C_{1}-P_{1}$. It is clear that $DT_{i}$ is the algebraic deviation of the completion time of job $i$  from the due date and $ES$ is the maximum possible shift (reduction of completion time) for the first job.

\begin{definition}\label{pl}
$PL$ is a vector of length $n$ and any element of $PL$ ($PL_{i}$) is the penalty possessed by job $i$. We define $PL$, as
\begin{equation}
PL_{i} =
\begin{cases}
	-\alpha_{i}, & \mbox{if }DT_{i} \leq 0 \\
	\beta_{i}, & \mbox{if }DT_{i} > 0 \;.
\end{cases}
\end{equation}
\end{definition}

With the above Definition we can express the objective function stated by Equation~\eqref{ob} as $\min(Sol)$, where
\begin{equation}\label{obn}
Sol = \sum\limits_{i=1}^{n} (DT_{i} \cdot PL_{i}) \;.
\end{equation}

The Algorithm~\ref{main} mentioned below returns the optimal solution value for any job sequence for the CDD problem on a single machine.

\begin{algorithm}
\DontPrintSemicolon
\textbf{Initialize} $C_{i}$ $\forall$ $i$ (Equation~\eqref{initial})\;
\textbf{Compute} $PL,DT,ES$\;
$Sol \gets \sum\limits_{i=1}^{n} (DT_{i} \cdot PL_{i})$\;
$j \gets 2$\;
\While{$(j < n+1)$}{
$C_{i} \gets C_{i}-\min\{ES,DT_{j}\}$, $\forall$ $i$\;
\textbf{Update} $PL,DT,ES$\;
$V_{j} \gets \sum\limits_{i=1}^{n} (DT_{i} \cdot PL_{i})$\;
\lIf {$(V_{j} < Sol)$}
	{ $Sol \gets V_{j}$}
\lElse{
	\Goto {return}\;
}
$j \gets j+1$\;
}
\Return $Sol$ \label{return} \;
\caption{Exact Algorithm for Single Machine}
\label{main}
\end{algorithm}

\section{Parallel Machine Case}
For the parallel machine case we first need to assign the jobs to each machine to get the number of jobs and their sequence in each machine. In addition to the parameters explained in Section~\ref{probform}, we define a new parameter $\lambda$, which is the machine assigned to each job.

\begin{definition}
We define $\lambda$ as the machine which has the earliest scheduled completion time of the last job on that machine. Using the notation mentioned in Section~\ref{probform}, $\lambda$ can be mathematically expressed as
\begin{equation*}
\lambda = \argmin_{j=1,2,\dots,m} M_j \;.\\
\end{equation*}
\label{lambda}
\end{definition}
\begin{algorithm}
\DontPrintSemicolon
$M_{j} \gets 0$ $\forall j=1,2,\dots,m$\;
$n_j \gets 1$ $\forall j=1,2,\dots,m$\;
$i \gets 0$\;
\For {$j \gets 1$ \normalfont\textbf{to} $m$}{
$i \gets i+1$\;
$W_{j}^{1} \gets i$\;
$M_{j} \gets \max\{P_{i},D\}$\;
}

\For {$i \gets m+1$ \normalfont\textbf{to} $n$}{
\textbf{Compute} $\lambda$\;
$n_{\lambda} \gets n_{\lambda}+1$\;
$W_{\lambda}^{n_{\lambda}} \gets i$\;
$M_{\lambda} \gets M_{\lambda}+P_{i}$\;
}
\For {each machine}{
\textbf{Algorithm~\ref{main}}\;
}
\caption{Exact Algorithm: Parallel Machine}
\label{parallel}
\end{algorithm}

Algorithm~\ref{parallel} assigns the first $m$ jobs to each machine respectively such that they all finish processing at the due date or after their processing time, whichever is higher. For the remaining jobs, we assign a machine $\lambda$ to job $i$ since it offers the least possible tardiness. Likewise each job is assigned at a specific machine such that the tardiness for all the jobs is the least for the given job sequence. The job sequence is maintained in the sense that for any two jobs $i$ and $j$ such that job $j$ follows $i$; the Algorithm~\ref{parallel} will either maintain this sequence or assign the same starting times at different machines to both the jobs. Finally, Algorithm~\ref{parallel} will give us the number of jobs $(n_{j})$ to be processed by any machine $j$ and the sequence of jobs in each machine, $W_{j}^{k}$. This is the best assignment of jobs at machines for the given sequence. Note that the sequence of jobs is still maintained here, since Algorithm~\ref{parallel} ensures that any job $i$ is not processed after a job $i+1$. Once we have the jobs assigned to each machine, the problem then converts to $m$ single machine problems, since all the machines are independent.

For the non-identical parallel machine case we need a slight change in the definition of $\lambda$ in Definition~\ref{lambda}. Recall that $M_j$ is the time at which machine $j$ finished its latest scheduled job and $\lambda$ is the machine which has the least completion time of jobs, among all the machines. In the non-identical machine case we need to make sure that the assigned machine not only has the least completion time but it is also feasible for the particular job(s). Hence, for the non-identical machines case, the definition of $\lambda$ in Algorithm~\ref{parallel} will change to $\lambda_i$ where

\begin{equation*}
\lambda_i = \argmin_{j=1,2,\dots,m} M_j\;, \text{ such that machine } j \text{ is feasible for job }i\;.\\
\end{equation*}

For the remaining part, the Algorithm~\ref{parallel} works in the same manner as for the identical parallel machines. Algorithm~\ref{parallel} can then be applied to the non-identical independent parallel machine case for the initial allocation of jobs to machines.

\section{Illustration of the Algorithms}
In this Section we explain Algorithm~\ref{main} and~\ref{parallel} with the help of illustrative examples consisting of $n=5$ jobs for both, single and parallel machine case. We optimize the given sequence of jobs $J$ where $J_i = i,$ $i=1,2,\dots,5$. The data for this example is given in Table~\ref{data}. There are five jobs to be processed against a common due-date ($D$) of $16$. The objective is to minimize Equation~\eqref{obn}.

\begin{table}
\centering
\caption{The data for the exemplary case. The parameters possess the same meaning as explained in Section~\ref{probform}.}
\label{data}
\begin{tabular}{ p{0.3cm}cp{0.3cm} p{0.3cm}cp{0.3cm} p{0.3cm}cp{0.3cm} p{0.3cm}cp{0.3cm} }\\ \hline\noalign{\smallskip}
&$i$ &&& $P_{i}$ &&& $\alpha_i$ &&& $\beta_i$ & \\ \hline
&1 &&& 6 &&& 7 &&& 9 &\\
&2 &&& 5 &&& 9 &&& 5 &\\
&3 &&& 2 &&& 6 &&& 4 &\\
&4 &&& 4 &&& 9 &&& 3 &\\
&5 &&& 4 &&& 3 &&& 2 &\\
\noalign{\smallskip}\hline\noalign{\smallskip}
\end{tabular}
\end{table}

\subsection{Single Machine Case}
We first initialize the completion times of all the jobs according to Equation~\eqref{initial} as shown in Figure~\ref{sm1}. The first job is completed at the due-date and possesses no penalty. However, all the remaining jobs from $J_i, i=2,3,4,5$ are tardy. After the initialization, the total penalty of this schedule is $Sol = \sum_{\substack{i=1}}^{n} (\alpha_{i}\cdot E_i + \beta_{i}\cdot T_i)=(0 \cdot 7+ 0\cdot 9) + (0 \cdot 9 + 5 \cdot 5) + (0 \cdot 6 + 7 \cdot 4) + (0 \cdot 9 + 11 \cdot 3) + (0 \cdot 3 + 15 \cdot 2)$. Hence, the objective value $Sol = 116$.

\begin{figure}[bht]
\centering
\begin{tikzpicture}[thick]
\filldraw[fill=gray!50,thick] (-3.5,0) rectangle (-1.4,0.7) node[midway] {6};
\filldraw[fill=gray!20,thick] (-1.4,0) rectangle (0.35,0.7) node[midway] {5}; 
\filldraw[fill=gray,thick] (0.35,0) rectangle (1.05,0.7) node[midway] {2};
\filldraw[fill=gray!20,thick] (1.05,0) rectangle (2.45,0.7) node[midway] {4};
\filldraw[fill=gray!50,thick] (2.45,0) rectangle (3.85,0.7) node[midway] {4};
\draw[thick,->] (-7,0) -- (4.6,0);
\draw[thick,snake=ticks,segment length=0.7cm,] (-7,0) -- (4.6,0);
\node at (4.6 , -0.25) {$t$};
\node at (-7 , -0.25) {$0$};
\node at (-7+0.7 , -0.25) {$2$};
\node at (-7+0.7*2 , -0.25) {$4$};
\node at (-7+0.7*3 , -0.25) {$6$};
\node at (-7+0.7*4 , -0.25) {$8$};
\node at (-7+0.7*5 , -0.25) {$10$};
\node at (-7+0.7*6 , -0.25) {$12$};
\node at (-7+0.7*7 , -0.25) {$14$};
\node at (-7+0.7*8 , -0.25) {$16$};
\node at (-7+0.7*9 , -0.25) {$18$};
\node at (-7+0.7*10 , -0.25) {$20$};
\node at (-7+0.7*11 , -0.25) {$22$};
\node at (-7+0.7*12 , -0.25) {$24$};
\node at (-7+0.7*13 , -0.25) {$26$};
\node at (-7+0.7*14 , -0.25) {$28$};
\node at (-7+0.7*15 , -0.25) {$30$};
\node at (-7+0.7*16 , -0.25) {$32$};
\end{tikzpicture}
\caption{Initialization of the completion times of all the jobs. The first job completes processing at the due date and the remaining jobs follow without any idle time.}
\label{sm1}
\end{figure}
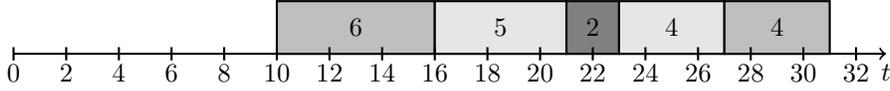

After the first left shift of $5$ time units, the total penalty of this schedule is $Sol = \sum_{\substack{i=1}}^{n} (\alpha_{i}\cdot E_i + \beta_{i}\cdot T_i)=(5 \cdot 7+ 0\cdot 9) + (0 \cdot 9 + 0 \cdot 5) + (0 \cdot 6 + 2 \cdot 4) + (0 \cdot 9 + 6 \cdot 3) + (0 \cdot 3 + 10 \cdot 2)$. Hence the objective value $Sol = 81$.

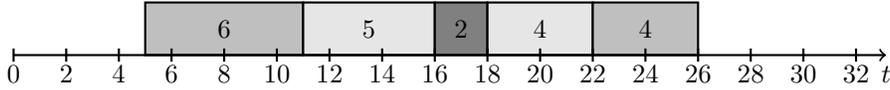
\begin{figure}[bht]
\centering
\begin{tikzpicture}[thick]
\filldraw[fill=gray!50,thick] (-7+2.5*0.7,0) rectangle (-7+5.5*0.7,0.7) node[midway] {6};
\filldraw[fill=gray!20,thick] (-7+5.5*0.7,0) rectangle (-7+8*0.7,0.7) node[midway] {5}; 
\filldraw[fill=gray,thick] (-7+8*0.7,0) rectangle (-7+9*0.7,0.7) node[midway] {2};
\filldraw[fill=gray!20,thick] (-7+9*0.7,0) rectangle (-7+11*0.7,0.7) node[midway] {4};
\filldraw[fill=gray!50,thick] (-7+11*0.7,0) rectangle (-7+13*0.7,0.7) node[midway] {4};
\draw[thick,->] (-7,0) -- (4.6,0);
\draw[thick,snake=ticks,segment length=0.7cm,] (-7,0) -- (4.6,0);
\node at (4.6 , -0.25) {$t$};
\node at (-7 , -0.25) {$0$};
\node at (-7+0.7 , -0.25) {$2$};
\node at (-7+0.7*2 , -0.25) {$4$};
\node at (-7+0.7*3 , -0.25) {$6$};
\node at (-7+0.7*4 , -0.25) {$8$};
\node at (-7+0.7*5 , -0.25) {$10$};
\node at (-7+0.7*6 , -0.25) {$12$};
\node at (-7+0.7*7 , -0.25) {$14$};
\node at (-7+0.7*8 , -0.25) {$16$};
\node at (-7+0.7*9 , -0.25) {$18$};
\node at (-7+0.7*10 , -0.25) {$20$};
\node at (-7+0.7*11 , -0.25) {$22$};
\node at (-7+0.7*12 , -0.25) {$24$};
\node at (-7+0.7*13 , -0.25) {$26$};
\node at (-7+0.7*14 , -0.25) {$28$};
\node at (-7+0.7*15 , -0.25) {$30$};
\node at (-7+0.7*16 , -0.25) {$32$};
\end{tikzpicture}
\caption{All the jobs are shifted left by $\min\{ES,DT_j\}=5$ units processing time. }
\label{sm2}
\end{figure}

After the third left shift of $2$ time units (Figure~\ref{sm3}), the total penalty of this schedule is $Sol = \sum_{\substack{i=1}}^{n} (\alpha_{i}\cdot E_i + \beta_{i}\cdot T_i)=(7 \cdot 7+ 0\cdot 9) + (2 \cdot 9 + 0 \cdot 5) + (0 \cdot 6 + 0 \cdot 4) + (0 \cdot 9 + 4 \cdot 3) + (0 \cdot 3 + 4 \cdot 2)$. Hence the objective value $Sol = 95$.

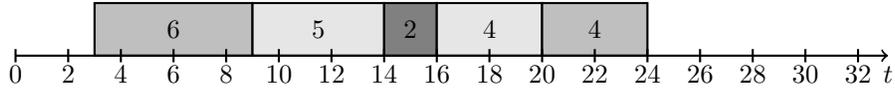
\begin{figure}[bht]
\centering
\begin{tikzpicture}[thick]
\filldraw[fill=gray!50,thick] (-7+1.5*0.7,0) rectangle (-7+4.5*0.7,0.7) node[midway] {6};
\filldraw[fill=gray!20,thick] (-7+4.5*0.7,0) rectangle (-7+7*0.7,0.7) node[midway] {5}; 
\filldraw[fill=gray,thick] (-7+7*0.7,0) rectangle (-7+8*0.7,0.7) node[midway] {2};
\filldraw[fill=gray!20,thick] (-7+8*0.7,0) rectangle (-7+10*0.7,0.7) node[midway] {4};
\filldraw[fill=gray!50,thick] (-7+10*0.7,0) rectangle (-7+12*0.7,0.7) node[midway] {4};
\draw[thick,->] (-7,0) -- (4.6,0);
\draw[thick,snake=ticks,segment length=0.7cm,] (-7,0) -- (4.6,0);
\node at (4.6 , -0.25) {$t$};
\node at (-7 , -0.25) {$0$};
\node at (-7+0.7 , -0.25) {$2$};
\node at (-7+0.7*2 , -0.25) {$4$};
\node at (-7+0.7*3 , -0.25) {$6$};
\node at (-7+0.7*4 , -0.25) {$8$};
\node at (-7+0.7*5 , -0.25) {$10$};
\node at (-7+0.7*6 , -0.25) {$12$};
\node at (-7+0.7*7 , -0.25) {$14$};
\node at (-7+0.7*8 , -0.25) {$16$};
\node at (-7+0.7*9 , -0.25) {$18$};
\node at (-7+0.7*10 , -0.25) {$20$};
\node at (-7+0.7*11 , -0.25) {$22$};
\node at (-7+0.7*12 , -0.25) {$24$};
\node at (-7+0.7*13 , -0.25) {$26$};
\node at (-7+0.7*14 , -0.25) {$28$};
\node at (-7+0.7*15 , -0.25) {$30$};
\node at (-7+0.7*16 , -0.25) {$32$};
\end{tikzpicture}
\caption{All the jobs are shifted left by $\min\{ES,DT_j\}=2$ units processing time.}
\label{sm3}
\end{figure}

Since the new value of the objective function is higher than in the previous step, we have the optimal value and schedule for this problem as shown in Figure~\ref{sm2} with a total penalty of $81$.

\begin{figure}[bht]
\centering
\begin{tikzpicture}[thick]
\filldraw[fill=gray!50,thick] (-7,0) rectangle (-7+3*0.7,0.7) node[midway] {6};
\filldraw[fill=gray!20,thick] (-7+3*0.7,0) rectangle (-7+5.5*0.7,0.7) node[midway] {5}; 
\filldraw[fill=gray,thick] (-7+5.5*0.7,0) rectangle (-7+6.5*0.7,0.7) node[midway] {2};
\filldraw[fill=gray!20,thick] (-7+6.5*0.7,0) rectangle (-7+8.5*0.7,0.7) node[midway] {4};
\filldraw[fill=gray!50,thick] (-7+8.5*0.7,0) rectangle (-7+10.5*0.7,0.7) node[midway] {4};
\draw[thick,->] (-7,0) -- (4.6,0);
\draw[thick,snake=ticks,segment length=0.7cm,] (-7,0) -- (4.6,0);
\node at (4.6 , -0.25) {$t$};
\node at (-7 , -0.25) {$0$};
\node at (-7+0.7 , -0.25) {$2$};
\node at (-7+0.7*2 , -0.25) {$4$};
\node at (-7+0.7*3 , -0.25) {$6$};
\node at (-7+0.7*4 , -0.25) {$8$};
\node at (-7+0.7*5 , -0.25) {$10$};
\node at (-7+0.7*6 , -0.25) {$12$};
\node at (-7+0.7*7 , -0.25) {$14$};
\node at (-7+0.7*8 , -0.25) {$16$};
\node at (-7+0.7*9 , -0.25) {$18$};
\node at (-7+0.7*10 , -0.25) {$20$};
\node at (-7+0.7*11 , -0.25) {$22$};
\node at (-7+0.7*12 , -0.25) {$24$};
\node at (-7+0.7*13 , -0.25) {$26$};
\node at (-7+0.7*14 , -0.25) {$28$};
\node at (-7+0.7*15 , -0.25) {$30$};
\node at (-7+0.7*16 , -0.25) {$32$};
\end{tikzpicture}
\caption{Final left shift by $ES=3$ units.}
\label{sm4}
\end{figure}
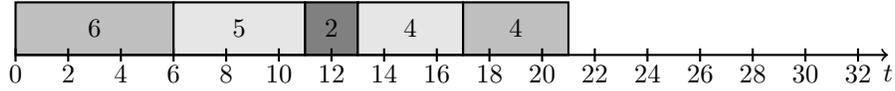

For the sake of completeness, Figure~\ref{sm4} shows the next step if we continue reducing the completion times using the same criterion as before. After the last possible left shift of $3$ time units, the total penalty of this schedule is $Sol = \sum_{\substack{i=1}}^{n} (\alpha_{i}\cdot E_i + \beta_{i}\cdot T_i)=(10 \cdot 7+ 0\cdot 9) + (5 \cdot 9 + 0 \cdot 5) + (3 \cdot 6 + 0 \cdot 4) + (0 \cdot 9 + 1 \cdot 3) + (0 \cdot 3 + 5 \cdot 2)$. Hence the objective value $Sol = 146$. The total penalty increases further to a value of $146$. Hence, the optimal value for this sequence is $81$.

\subsection{Parallel Machine Case}
In the parallel machine case we consider two parallel machines and illustrate how we first assign the jobs in the same job sequence $J$ to the machines and optimize them independently. The data used in this example is the same as in Table~\ref{data}. The common due-date for the instance is also the same as earlier, $D=16$.

\begin{figure}[bht]
\centering
\subfigure[ ]{\begin{tikzpicture}[thick,framed]
\draw[thick,->] (-7,1) -- (-7+0.45*11.5,1);
\draw[thick,snake=ticks,segment length=0.45cm] (-7,1) -- (-7+0.45*11.5,1);
\node at (-7 , 0.75) {{\scriptsize $0$}};
\node at (-7+0.45 , 0.75) {{\scriptsize $2$}};
\node at (-7+0.45*2 , 0.75) {{\scriptsize $4$}};
\node at (-7+0.45*3 , 0.75) {{\scriptsize $6$}};
\node at (-7+0.45*4 , 0.75) {{\scriptsize $8$}};
\node at (-7+0.45*5 , 0.75) {{\scriptsize $10$}};
\node at (-7+0.45*6 , 0.75) {{\scriptsize $12$}};
\node at (-7+0.45*7 , 0.75) {{\scriptsize $14$}};
\node at (-7+0.45*8 , 0.75) {{\scriptsize $16$}};
\node at (-7+0.45*9 , 0.75) {{\scriptsize $18$}};
\node at (-7+0.45*10 , 0.75) {{\scriptsize $20$}};
\node at (-7+0.45*11 , 0.75) {{\scriptsize $22$}};
\node at (-7+0.45*11.5 , 1.25) {{\scriptsize $t$}};
\draw[thick,->] (-7,0) -- (-7+0.45*11.5,0);
\draw[thick,snake=ticks,segment length=0.45cm] (-7,0) -- (-7+0.45*11.5,0);
\node at (-7 , -0.25) {{\scriptsize $0$}};
\node at (-7+0.45 , -0.25) {{\scriptsize $2$}};
\node at (-7+0.45*2 , -0.25) {{\scriptsize $4$}};
\node at (-7+0.45*3 , -0.25) {{\scriptsize $6$}};
\node at (-7+0.45*4 , -0.25) {{\scriptsize $8$}};
\node at (-7+0.45*5 , -0.25) {{\scriptsize $10$}};
\node at (-7+0.45*6 , -0.25) {{\scriptsize $12$}};
\node at (-7+0.45*7 , -0.25) {{\scriptsize $14$}};
\node at (-7+0.45*8 , -0.25) {{\scriptsize $16$}};
\node at (-7+0.45*9 , -0.25) {{\scriptsize $18$}};
\node at (-7+0.45*10 , -0.25) {{\scriptsize $20$}};
\node at (-7+0.45*11 , -0.25) {{\scriptsize $22$}};
\node at (-7+0.45*11.5 , 0.25) {{\scriptsize $t$}};
\filldraw[fill=gray!50,thick] (-7+0.45-0.2,-1.5) rectangle (-7+4*0.45-0.2,-1) node[midway] {{\footnotesize 6}};
\filldraw[fill=gray!20,thick] (-7+4*0.45-0.2,-1.5) rectangle (-7+6.5*0.45-0.2,-1) node[midway] {{\footnotesize 5}}; 
\filldraw[fill=gray,thick] (-7+6.5*0.45-0.2,-1.5) rectangle (-7+7.5*0.45-0.2,-1) node[midway] {{\footnotesize 2}};
\filldraw[fill=gray!20,thick] (-7+7.5*0.45-0.2,-1.5) rectangle (-7+9.5*0.45-0.2,-1) node[midway] {{\footnotesize 4}};
\filldraw[fill=gray!50,thick] (-7+9.5*0.45-0.2,-1.5) rectangle (-7+11.5*0.45-0.2,-1) node[midway] {{\footnotesize 4}};
\end{tikzpicture}}
\subfigure[ ]{\begin{tikzpicture}[thick,framed]
\draw[thick,->] (-7,1) -- (-7+0.45*11.5,1);
\draw[thick,snake=ticks,segment length=0.45cm] (-7,1) -- (-7+0.45*11.5,1);
\node at (-7 , 0.75) {{\scriptsize $0$}};
\node at (-7+0.45 , 0.75) {{\scriptsize $2$}};
\node at (-7+0.45*2 , 0.75) {{\scriptsize $4$}};
\node at (-7+0.45*3 , 0.75) {{\scriptsize $6$}};
\node at (-7+0.45*4 , 0.75) {{\scriptsize $8$}};
\node at (-7+0.45*5 , 0.75) {{\scriptsize $10$}};
\node at (-7+0.45*6 , 0.75) {{\scriptsize $12$}};
\node at (-7+0.45*7 , 0.75) {{\scriptsize $14$}};
\node at (-7+0.45*8 , 0.75) {{\scriptsize $16$}};
\node at (-7+0.45*9 , 0.75) {{\scriptsize $18$}};
\node at (-7+0.45*10 , 0.75) {{\scriptsize $20$}};
\node at (-7+0.45*11 , 0.75) {{\scriptsize $22$}};
\node at (-7+0.45*11.5 , 1.25) {{\scriptsize $t$}};
\draw[thick,->] (-7,0) -- (-7+0.45*11.5,0);
\draw[thick,snake=ticks,segment length=0.45cm] (-7,0) -- (-7+0.45*11.5,0);
\node at (-7 , -0.25) {{\scriptsize $0$}};
\node at (-7+0.45 , -0.25) {{\scriptsize $2$}};
\node at (-7+0.45*2 , -0.25) {{\scriptsize $4$}};
\node at (-7+0.45*3 , -0.25) {{\scriptsize $6$}};
\node at (-7+0.45*4 , -0.25) {{\scriptsize $8$}};
\node at (-7+0.45*5 , -0.25) {{\scriptsize $10$}};
\node at (-7+0.45*6 , -0.25) {{\scriptsize $12$}};
\node at (-7+0.45*7 , -0.25) {{\scriptsize $14$}};
\node at (-7+0.45*8 , -0.25) {{\scriptsize $16$}};
\node at (-7+0.45*9 , -0.25) {{\scriptsize $18$}};
\node at (-7+0.45*10 , -0.25) {{\scriptsize $20$}};
\node at (-7+0.45*11 , -0.25) {{\scriptsize $22$}};
\node at (-7+0.45*11.5 , 0.25) {{\scriptsize $t$}};
\filldraw[fill=gray!50,thick] (-7+5*0.45,1) rectangle (-7+8*0.45,1.4) node[midway] {{\scriptsize 6}};
\filldraw[fill=gray!20,thick] (-7+5.5*.45,0) rectangle (-7+7*0.45,0.5) node[midway] {{\scriptsize 5}}; 
\filldraw[fill=gray,thick] (-7+6.5*0.45-0.2,-1.5) rectangle (-7+7.5*0.45-0.2,-1) node[midway] {{\scriptsize 2}};
\filldraw[fill=gray!20,thick] (-7+7.5*0.45-0.2,-1.5) rectangle (-7+9.5*0.45-0.2,-1) node[midway] {{\scriptsize 4}};
\filldraw[fill=gray!50,thick] (-7+9.5*0.45-0.2,-1.5) rectangle (-7+11.5*0.45-0.2,-1) node[midway] {{\scriptsize 4}};

\filldraw[dashed,fill=none,] (-7+0.45-0.2,-1.5) rectangle (-7+4*0.45-0.2,-1);
\filldraw[dashed,fill=none,] (-7+4*0.45-0.2,-1.5) rectangle (-7+6.5*0.45-0.2,-1); 

\end{tikzpicture}}
\subfigure[ ]{\begin{tikzpicture}[thick,framed]
\draw[thick,->] (-7,1) -- (-7+0.45*11.5,1);
\draw[thick,snake=ticks,segment length=0.45cm] (-7,1) -- (-7+0.45*11.5,1);
\node at (-7 , 0.75) {{\scriptsize $0$}};
\node at (-7+0.45 , 0.75) {{\scriptsize $2$}};
\node at (-7+0.45*2 , 0.75) {{\scriptsize $4$}};
\node at (-7+0.45*3 , 0.75) {{\scriptsize $6$}};
\node at (-7+0.45*4 , 0.75) {{\scriptsize $8$}};
\node at (-7+0.45*5 , 0.75) {{\scriptsize $10$}};
\node at (-7+0.45*6 , 0.75) {{\scriptsize $12$}};
\node at (-7+0.45*7 , 0.75) {{\scriptsize $14$}};
\node at (-7+0.45*8 , 0.75) {{\scriptsize $16$}};
\node at (-7+0.45*9 , 0.75) {{\scriptsize $18$}};
\node at (-7+0.45*10 , 0.75) {{\scriptsize $20$}};
\node at (-7+0.45*11 , 0.75) {{\scriptsize $22$}};
\node at (-7+0.45*11.5 , 1.25) {{\scriptsize $t$}};
\draw[thick,->] (-7,0) -- (-7+0.45*11.5,0);
\draw[thick,snake=ticks,segment length=0.45cm] (-7,0) -- (-7+0.45*11.5,0);
\node at (-7 , -0.25) {{\scriptsize $0$}};
\node at (-7+0.45 , -0.25) {{\scriptsize $2$}};
\node at (-7+0.45*2 , -0.25) {{\scriptsize $4$}};
\node at (-7+0.45*3 , -0.25) {{\scriptsize $6$}};
\node at (-7+0.45*4 , -0.25) {{\scriptsize $8$}};
\node at (-7+0.45*5 , -0.25) {{\scriptsize $10$}};
\node at (-7+0.45*6 , -0.25) {{\scriptsize $12$}};
\node at (-7+0.45*7 , -0.25) {{\scriptsize $14$}};
\node at (-7+0.45*8 , -0.25) {{\scriptsize $16$}};
\node at (-7+0.45*9 , -0.25) {{\scriptsize $18$}};
\node at (-7+0.45*10 , -0.25) {{\scriptsize $20$}};
\node at (-7+0.45*11 , -0.25) {{\scriptsize $22$}};
\node at (-7+0.45*11.5 , 0.25) {{\scriptsize $t$}};
\filldraw[fill=gray!50,thick] (-7+5*0.45,1) rectangle (-7+8*0.45,1.5) node[midway] {{\scriptsize 6}};
\filldraw[fill=gray!20,thick] (-7+5.5*.45,0) rectangle (-7+8*0.45,0.5) node[midway] {{\scriptsize 5}}; 
\filldraw[fill=gray,thick] (-7+8*0.45,1) rectangle (-7+9*0.45,1.5) node[midway] {{\scriptsize 2}};
\filldraw[fill=gray!20,thick] (-7+8*0.45,0) rectangle (-7+10*0.45,0.5) node[midway] {{\scriptsize 4}};
\filldraw[fill=gray!50,thick] (-7+9.5*0.45-0.2,-1.5) rectangle (-7+11.5*0.45-0.2,-1) node[midway] {{\scriptsize 4}};

\filldraw[dashed,fill=none,] (-7+0.45-0.2,-1.5) rectangle (-7+4*0.45-0.2,-1);
\filldraw[dashed,fill=none,] (-7+4*0.45-0.2,-1.5) rectangle (-7+6.5*0.45-0.2,-1); 
\filldraw[dashed,fill=none,] (-7+6.5*0.45-0.2,-1.5) rectangle (-7+7.5*0.45-0.2,-1);
\filldraw[dashed,fill=none,] (-7+7.5*0.45-0.2,-1.5) rectangle (-7+9.5*0.45-0.2,-1);
\filldraw[fill=gray!50,thick] (-7+9.5*0.45-0.2,-1.5) rectangle (-7+11.5*0.45-0.2,-1) node[midway] {{\footnotesize 4}};
\end{tikzpicture}}
\subfigure[ ]{\begin{tikzpicture}[thick,framed]
\draw[thick,->] (-7,1) -- (-7+0.45*11.5,1);
\draw[thick,snake=ticks,segment length=0.45cm] (-7,1) -- (-7+0.45*11.5,1);
\node at (-7 , 0.75) {{\scriptsize $0$}};
\node at (-7+0.45 , 0.75) {{\scriptsize $2$}};
\node at (-7+0.45*2 , 0.75) {{\scriptsize $4$}};
\node at (-7+0.45*3 , 0.75) {{\scriptsize $6$}};
\node at (-7+0.45*4 , 0.75) {{\scriptsize $8$}};
\node at (-7+0.45*5 , 0.75) {{\scriptsize $10$}};
\node at (-7+0.45*6 , 0.75) {{\scriptsize $12$}};
\node at (-7+0.45*7 , 0.75) {{\scriptsize $14$}};
\node at (-7+0.45*8 , 0.75) {{\scriptsize $16$}};
\node at (-7+0.45*9 , 0.75) {{\scriptsize $18$}};
\node at (-7+0.45*10 , 0.75) {{\scriptsize $20$}};
\node at (-7+0.45*11 , 0.75) {{\scriptsize $22$}};
\node at (-7+0.45*11.5 , 1.25) {{\scriptsize $t$}};
\draw[thick,->] (-7,0) -- (-7+0.45*11.5,0);
\draw[thick,snake=ticks,segment length=0.45cm] (-7,0) -- (-7+0.45*11.5,0);
\node at (-7 , -0.25) {{\scriptsize $0$}};
\node at (-7+0.45 , -0.25) {{\scriptsize $2$}};
\node at (-7+0.45*2 , -0.25) {{\scriptsize $4$}};
\node at (-7+0.45*3 , -0.25) {{\scriptsize $6$}};
\node at (-7+0.45*4 , -0.25) {{\scriptsize $8$}};
\node at (-7+0.45*5 , -0.25) {{\scriptsize $10$}};
\node at (-7+0.45*6 , -0.25) {{\scriptsize $12$}};
\node at (-7+0.45*7 , -0.25) {{\scriptsize $14$}};
\node at (-7+0.45*8 , -0.25) {{\scriptsize $16$}};
\node at (-7+0.45*9 , -0.25) {{\scriptsize $18$}};
\node at (-7+0.45*10 , -0.25) {{\scriptsize $20$}};
\node at (-7+0.45*11 , -0.25) {{\scriptsize $22$}};
\node at (-7+0.45*11.5 , 0.25) {{\scriptsize $t$}};
\filldraw[fill=gray!50,thick] (-7+5*0.45,1) rectangle (-7+8*0.45,1.5) node[midway] {{\scriptsize 6}};
\filldraw[fill=gray!20,thick] (-7+5.5*.45,0) rectangle (-7+8*0.45,0.5) node[midway] {{\scriptsize 5}}; 
\filldraw[fill=gray,thick] (-7+8*0.45,1) rectangle (-7+9*0.45,1.5) node[midway] {{\scriptsize 2}};
\filldraw[fill=gray!20,thick] (-7+8*0.45,0) rectangle (-7+10*0.45,0.5) node[midway] {{\scriptsize 4}};
\filldraw[fill=gray!50,thick] (-7+9*0.45,1) rectangle (-7+11*0.45,1.5) node[midway] {{\scriptsize 4}};

\filldraw[dashed,fill=none,] (-7+0.45-0.2,-1.5) rectangle (-7+4*0.45-0.2,-1);
\filldraw[dashed,fill=none,] (-7+4*0.45-0.2,-1.5) rectangle (-7+6.5*0.45-0.2,-1); 
\filldraw[dashed,fill=none,] (-7+6.5*0.45-0.2,-1.5) rectangle (-7+7.5*0.45-0.2,-1);
\filldraw[dashed,fill=none,] (-7+7.5*0.45-0.2,-1.5) rectangle (-7+9.5*0.45-0.2,-1);
\filldraw[dashed,fill=none,] (-7+9.5*0.45-0.2,-1.5) rectangle (-7+11.5*0.45-0.2,-1);

\end{tikzpicture}}
\caption{Illustration of the assignment of jobs to machines. After the assignment, each machine has a certain number of jobs in the given sequence.}
\label{pm}
\end{figure}
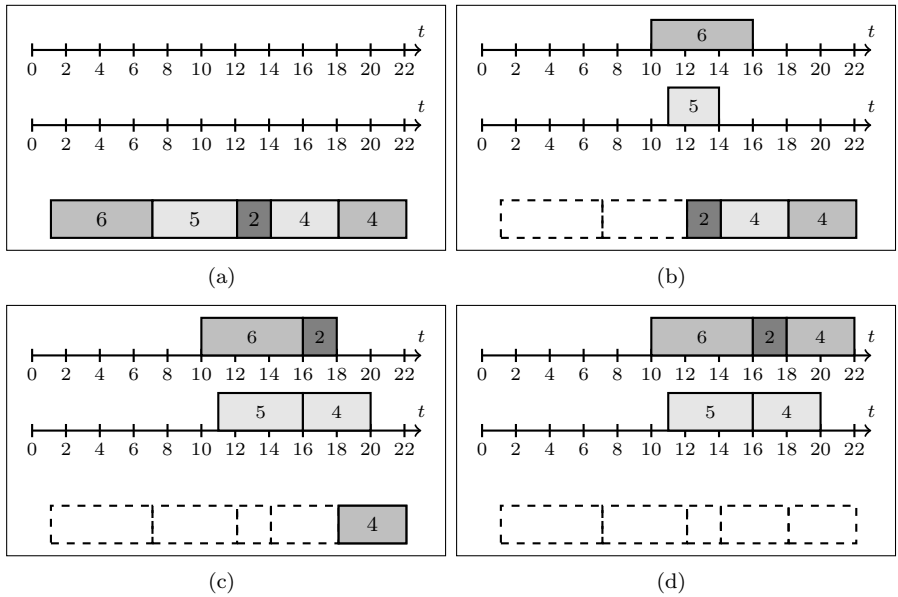

As shown in Figure~\ref{pm}(a), there are five jobs to be processed on two independent identical parallel machines, against a due-date of $16$. Hence, we first assign the jobs to a machine. We start with the first two jobs in the sequence $J$ and assign them to the machines separately at $\max\{P_i,D\}$, Figure~\ref{pm}(b). For the remaining jobs, we subsequently choose a machine which offers least tardiness for each job. The third job in the sequence is assigned to the first machine and the fourth job goes to the second machine on the same lines, as depicted in Figure~\ref{pm}(c). Finally, we have all the jobs assigned to a machine (Figure~\ref{pm}(d)) and each machine has a certain number of jobs to process in a given sequence. In this example, the first machine processes $3$ jobs with the processing times of $6,2$ and $4$, while the second machine processes $2$ jobs with processing times of $5$ and $4$, in that order. Once we have this assignment of jobs to machines, we can apply our single machine algorithm to both of them independently to optimize the overall earliness and tardiness penalty. Figure~\ref{popt} shows the best schedule for both the machines with an overall penalty of $32$.

\begin{figure}[bht]
\centering
\begin{tikzpicture}[thick,framed]
\draw[thick,->] (-7,2) -- (-7+0.7*11.5,2);
\draw[thick,snake=ticks,segment length=0.7cm] (-7,2) -- (-7+0.7*11.5,2);
\node at (-7 , 1.75) {{\scriptsize $0$}};
\node at (-7+0.7 , 1.75) {{\scriptsize $2$}};
\node at (-7+0.7*2 , 1.75) {{\scriptsize $4$}};
\node at (-7+0.7*3 , 1.75) {{\scriptsize $6$}};
\node at (-7+0.7*4 , 1.75) {{\scriptsize $8$}};
\node at (-7+0.7*5 , 1.75) {{\scriptsize $10$}};
\node at (-7+0.7*6 , 1.75) {{\scriptsize $12$}};
\node at (-7+0.7*7 , 1.75) {{\scriptsize $14$}};
\node at (-7+0.7*8 , 1.75) {{\scriptsize $16$}};
\node at (-7+0.7*9 , 1.75) {{\scriptsize $18$}};
\node at (-7+0.7*10 , 1.75) {{\scriptsize $20$}};
\node at (-7+0.7*11 , 1.75) {{\scriptsize $22$}};
\node at (-7+0.7*11.5 , 2.25) {{\scriptsize $t$}};
\draw[thick,->] (-7,0) -- (-7+0.7*11.5,0);
\draw[thick,snake=ticks,segment length=0.7cm] (-7,0) -- (-7+0.7*11.5,0);
\node at (-7 , -0.25) {{\scriptsize $0$}};
\node at (-7+0.7 , -0.25) {{\scriptsize $2$}};
\node at (-7+0.7*2 , -0.25) {{\scriptsize $4$}};
\node at (-7+0.7*3 , -0.25) {{\scriptsize $6$}};
\node at (-7+0.7*4 , -0.25) {{\scriptsize $8$}};
\node at (-7+0.7*5 , -0.25) {{\scriptsize $10$}};
\node at (-7+0.7*6 , -0.25) {{\scriptsize $12$}};
\node at (-7+0.7*7 , -0.25) {{\scriptsize $14$}};
\node at (-7+0.7*8 , -0.25) {{\scriptsize $16$}};
\node at (-7+0.7*9 , -0.25) {{\scriptsize $18$}};
\node at (-7+0.7*10 , -0.25) {{\scriptsize $20$}};
\node at (-7+0.7*11 , -0.25) {{\scriptsize $22$}};
\node at (-7+0.7*11.5 , 0.25) {{\scriptsize $t$}};
\filldraw[fill=gray!50,thick] (-7+5*0.7,2) rectangle (-7+8*0.7,3) node[midway] {{\scriptsize 6}};
\filldraw[fill=gray!20,thick] (-7+5.5*.7,0) rectangle (-7+8*0.7,1) node[midway] {{\scriptsize 5}}; 
\filldraw[fill=gray,thick] (-7+8*0.7,2) rectangle (-7+9*0.7,3) node[midway] {{\scriptsize 2}};
\filldraw[fill=gray!20,thick] (-7+8*0.7,0) rectangle (-7+10*0.7,1) node[midway] {{\scriptsize 4}};
\filldraw[fill=gray!50,thick] (-7+9*0.7,2) rectangle (-7+11*0.7,3) node[midway] {{\scriptsize 4}};

\end{tikzpicture}
\caption{Final optimal schedule for both the machines for the given sequence of jobs. The overall penalty of $32$ is reached, which is the best solution value as per Algorithm~\ref{main} and~\ref{parallel}.}
\label{popt}
\end{figure}
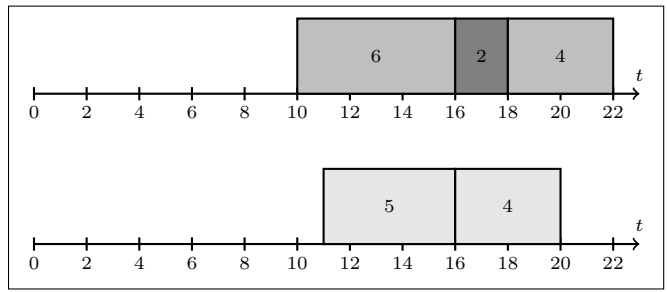

\section{Proof of Optimality}

We now prove the optimality of Algorithm~\ref{main} with respect to the solution value for the single machine case.

\begin{lemma}\label{lemma}
If the initial assignment of the completion times of the jobs ($C_i$), for a given sequence $J$ is done according to Equation~\eqref{initial}, then the optimal solution for this sequence can be obtained only by reducing the completion times of all the jobs or leaving them unchanged.
\end{lemma}

\begin{proof}
We prove the above lemma by considering two cases of Equation~\eqref{initial}. \\
\noindent
\textbf{Case 1:} $D > P_{1}$ \\
In this case Equation~\eqref{initial} will ensure that the first job is completed at the due-date and the following jobs are processed consecutively without any idle time. Moreover, with this assignment all the jobs will be tardy except for the first job which will be completed at the due date. The total penalty (say, $PN$) will be $\sum_{\substack{i=1}}^{n} (\beta_{i} \cdot T_{i})$, where $T_{i}=C_{i}-D$, $i=1,2,\dots,n$. Now if we increase the completion time of the first job by $x$ units then the new completion times $C_{i}^{\prime}$ for the jobs will be $C_{i} + x$ $\forall$ $i,(i=1,2,\dots,n)$ and the new total penalty $PN^{\prime}$ will be $\sum_{\substack{i=1}}^{n} (\beta_{i} \cdot T_{i}^{'})$, where $T_{i}^{'}=T_{i}+x$ $(i=1,2,\dots,n)$. Clearly, we have $PN^{\prime} > PN$ which proves that an increase in the completion times cannot fetch optimality which in turn proves that optimality can be achieved only by reducing the completion times or leaving them unchanged from Equation~\eqref{initial}.\\
\noindent
\textbf{Case 2:} $D \leq P_{1}$ \\
If the processing time of the first job in any given sequence is more than the due-date then all the jobs will be tardy including the first job as $P_1>D$. Since all the jobs are already tardy, a right shift ({\em i.e.\ \/}increasing the completion times) of the jobs will only increase the total penalty and hence worsening the solution. Moreover, a left shift ({\em i.e.\ \/}reducing the completion times) of the jobs is not possible either, because $C_{1}=P_{1}$, which means that the first job will start at time $0$. Hence, in such a case Equation~\eqref{initial} is the optimal solution. In the rest of the paper we avoid this simple case and assume that for any given sequence the processing time of the first job is less than the due-date.
\end{proof}

\begin{theorem}\label{theorem}
Algorithm~\ref{main} finds the optimal solution for a single machine common due date problem, for a given job sequence.
\end{theorem}

\begin{proof}
The initialization of the completion times for a sequence $P$ is done according to Lemma~\ref{lemma}.  It is evident from Equation~\eqref{initial} that the deviation from the due date ($DT_{i}$) is zero for the first job and greater than zero for all the following jobs. Besides, $DT_i < DT_{i+1}$ for $i=1,2,3,\dots,n-1$, since $C_{i}<C_{i+1}$ from Equation~\eqref{initial} and $DT_{i}$ is defined as $DT_{i}=C_{i}-D$. By Lemma~\ref{lemma} the optimal solution for this sequence can be achieved only by reducing the completion times of all the jobs simultaneously or leaving the completion times unchanged. Besides, a reduction of the completion times is possible only if $ES>0$ since there is no idle time between any jobs.

The total penalty after the initialization is $PN=\sum_{\substack{i=1}}^{n} (\beta_{i} \cdot T_{i})$ since none of the jobs are completed before the due date. According to Algorithm~\ref{main} the completion times of all the jobs is reduced by $\min\{ES,DT_{j}\}$ at any iteration. Since $DT_1=0$, there will be no loss or gain for $j=1$. After any iteration of the $while$ loop in line $5$, we decrease the total weighted tardiness but gain some weighted earliness penalty for some jobs. A reduction of the completion times by $\min\{ES,DT_{j}\}$ is the best non-greedy reduction. Let $\min\{ES,DT_{j}\}>0$ and $t$ be a number between $0$ and $\min\{ES,DT_{j}\}$. Then reducing the completion times by $t$ will increase the number of early jobs by one and reduce the number of tardy jobs by one. With this operation, if there is an improvement to the overall solution, then a reduction by $\min\{ES,DT_{j}\}$ will fetch a much better solution ($V_{j}$) because reducing the completion times by $t$ will lead to a situation where none of the jobs either start at time $0$ (because $ES>0$) nor any of the jobs finish at the due date since the jobs $1,2,3,\dots,j-1$ are early, jobs $j,j+1,\dots,n$ are tardy and the new completion time of job $j$ is $C_{j}^{'}=C_{j}-t$. 

Since after this reduction $DT_{j}>0$ and $DT_{j}<DT_{j+1}$ for $j=1,2,3,\dots,n-1$, none of the jobs will finish at the due date after a reduction by $t$ units. Moreover, it was proved by Cheng \emph{et al.}~\cite{cheng} that in an optimal schedule for the restrictive common due date, either one of the jobs should start at time $0$ or one of the jobs should end at the due date. This case can occur only if we reduce the completion times by $\min\{ES,DT_{j}\}$. If $ES<DT_{j}$, the first job will start at time $0$ and if $DT_j<ES$ then one of the jobs will end at the due date. In the next iterations we continue the reductions as long as we get an improvement in the solution and once the new solution is not better than the previous best, we do not need to check any further and we have our optimal solution. This can be proved by considering the values of the objective function at the indices of two iterations; $j$ and $j+1$. Let $V_{j}$ and $V_{j+1}$ be the value of the objective function at these two indices, then the solution cannot be improved any further if $V_{j+1}>V_{j}$ by Lemma~\ref{lemma3}. \end{proof}

\begin{lemma}\label{lemma3}
Once the value of the solution at any iteration $j$ is less than the value at iteration $j+1$, the solution cannot be improved any further.

\end{lemma}
\begin{proof}
If $V_{j+1}>V_{j}$, a further left shift of the jobs does not fetch a better solution. Note that the objective function has two parts: penalty due to earliness and penalty due to tardiness. Let us consider the earliness and tardiness of the jobs after the $j$th iterations are $E_{i}^{j}$ and $T_{i}^{j}$ for $i=1,2,\dots,n$. Then we have $V_{j}=\sum_{\substack{i=1}}^{n} (\alpha_{i}E_{i}^{j}+\beta_{i}T_{i}^{j})$ and $V^{j+1}=\sum_{\substack{i=1}}^{n} (\alpha_{i}E_{i}^{j+1}+\beta_{i}T_{i}^{j+1})$. Besides, after every iteration of the $while$ loop in Algorithm~\ref{main}, the completion times are reduced or in other words the jobs are shifted left. This leads to an increase in the earliness and a decrease in the tardiness of the jobs. Let's say, the difference in the reduction between $V^{j+1}$ and $V^{j}$ is $x$. Then we have $E^{j+1}=E^{j}+x$ and $T_{j+1}=T_{j}-x$. Since $V^{j+1}>V^{j}$, we have: $\sum_{\substack{i=1}}^{n} (\alpha_{i}E_{i}^{j+1}+\beta_{i}T_{i}^{j+1}) > \sum_{\substack{i=1}}^{n} (\alpha_{i}E_{i}^{j}+\beta_{i}T_{i}^{j})$. By substituting the values of $E^{j+1}$ and $T^{j+1}$ we get,  $\sum_{\substack{i=1}}^{j+1} \alpha_{i}x > \sum_{\substack{i=j+2}}^{n} \beta_{i}x$. Hence, at the $(j+1)^{th}$ iteration the total penalty due to earliness exceeds the total penalty due to tardiness. This proves that for any further reduction there cannot be an improvement in the solution because a decrease in the tardiness penalty will always be less than the increase in the earliness penalty. 
\end{proof}

\section{Algorithm Run-Time Complexity} 
In this Section we study and prove the run-time complexity of the Algorithms~\ref{main} and~\ref{parallel}. We calculate the complexities of all the algorithms separately considering the worst cases for all. Let $T_1$ and $T_2$ be the run-time complexities of the algorithms respectively.

\begin{lemma}
The run-time complexities of both Algorithms~\ref{main} and~\ref{parallel} are $O(n^2)$, where $n$ is the total number of jobs.
\end{lemma}

\begin{proof}

As for Algorithm~\ref{main}, the calculations involved in the initialization step and evaluation of $PL,DT,ES,Sol$ are all of $O(n)$ complexity and their evaluation is irrespective of the any conditions unlike inside the $while$ loop. The $while$ loop again evaluates and updates these parameters at every step of its iteration and returns the output once their is no improvement possible. The worst case will occur when the $while$ loop is iterated over all the values of $j$, $j=2,3,\dots,n$. Hence the complexity of Algorithm~\ref{main} is $O(n^2)$ with $n$ being the number of jobs processed by the machine. Hence, $T_1=O(n^2)$.

Let $m$ be the number of machines, then in the Algorithm~\ref{parallel}, the complexity for the first two \textit{for} loops is $O(m + (n-m)m)$ where, $O(m)$ corresponds to the first \textit{for} loop and $O((n-m)m)$ corresponds to the second \textit{for} loop involving the calculation of $\lambda$. For the last \textit{for} loop, we need to consider all the cases of the number of jobs processed by each machine. 

Let $x_1,x_2,x_3,\dots,x_m$ be the number of jobs processed by the machines, respectively. Then, $\sum_{\substack{i=1}}^{m} x_i=n$. We make a reasonable assumption that the number of machines is less than the number of jobs, which is usually the case. In such a case the complexity of Algorithm~\ref{parallel} ($T_2$) is equal to  $O(m+ nm - m^2)+ \sum_{\substack{i=1}}^{m}O(x_{i}^{2})$. Since $\sum_{\substack{i=1}}^{m} x_{i} = n$, we have $\sum_{\substack{i=1}}^{m}O(x_{i}^{2})=O(n^2)$. Thus the complexity of Algorithm~\ref{parallel} is $O(m+nm-m^2+n^2)$. Since we assume $m<n$ we have $T_2 = O(n^2)$.
\end{proof}

\section{Exponential Search: An Efficient Implementation of Algorithm~\ref{main}}
Algorithm~\ref{main} shifts the jobs to the left by reducing the completion times of all the jobs by $\min\{ES,DT_{j}\}$ on every iteration of the $while$ loop. The runtime complexity of the algorithm can be improved form $O(n^2)$ to $O(n\log n)$ by implementing an exponential search instead of a step by step reduction, as in Algorithm~\ref{main}. To explain this we first need to understand the slope of the objective function values for each iteration. In the proof of optimality of Algorithm~\ref{main}, we proved that there is only one minimum present in $V^{j}$ $\forall j$. Besides, the value of $DT_{j}$ increases for every $j$ as it depends on the completion times. Also note that the reduction in the completion times is made by $\min\{ES,DT_{j}\}$. Hence, if for any $j$, $ES \leq DT_{j}$ then every iteration after $j$ will fetch the same objective function value, $V^{j}$. Hence, the solution values after each iteration will have a trend as shown below in Figure~\ref{trend}.

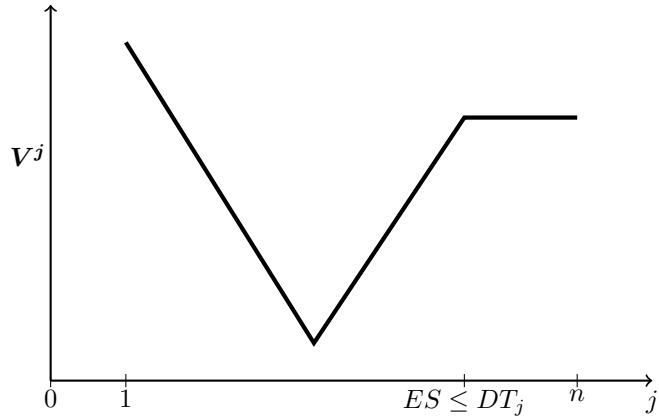
\begin{figure}[ht]
\centering
\begin{tikzpicture}
\draw[thick,->] (0,0) -- (8,0) node[anchor=north] {$j$};
\draw[snake=ticks,segment length=1cm ] (0,0) -- (1,0);
\draw[snake=ticks,segment length=1cm ] (5.5,0) -- (5.5,0);
\draw[snake=ticks,segment length=1cm ] (7,0) -- (7,0);
\draw	(0,0) node[anchor=north] {$0$}
		(1,0) node[anchor=north] {$1$}
		(5.5,0) node[anchor=north] {$ES \leq DT_j $}
		(7,0) node[anchor=north] {$n$};
\draw[thick,->] (0,0) -- (0,5);
\draw	(-0.3,3) node{$\boldsymbol{V^j}$};
\draw[ultra thick] (1,4.5) -- (3.5,0.5) -- (5.5,3.5) -- (7,3.5);

\end{tikzpicture}
\caption{The trend of the solution value against each iteration of Algorithm~\ref{main}, for a job sequence. The value of the solution does not improve any further after a certain number of reductions.}
\label{trend}
\end{figure}

With such a slope of the solution we can use the exponential search as opposed to a step by step search, which will in turn improve the run-time complexity of Algorithm~\ref{main}. This can be achieved by increasing or decreasing the step size of the $while$ loop by orders of $2$ ({\em i.e.\ \/}$2,2^2,2^3,\dots,n$) while keeping track of the slope of the solution. The index of the next iteration should be increased if the slope is negative and decreased if the slope is non-negative. At each step we need to keep track of the previous two indices and once the difference between the indices is less than the minimum of the two, then we need to perform binary search on the same lines. The optimum will be reached if both the adjacent solutions are greater than the current value. In this methodology we do not need to search for all values of $j$ but in steps of $2^j$. Hence the run-time complexity with exponential search will be $O(n\log n)$ for both the single machine and parallel machine cases.

\section{A Dynamic Case of CDD}
In this Section we discuss about a dynamic case of the common due-date problem for the single machine case at the planning stage. Consider the case when an optimal schedule has been calculated for a certain number jobs, and then an unknown number of jobs with unknown processing times arrive later. We assume that the original schedule is not disturbed and the new sequence of jobs can be processed after the first set of jobs. We show that in such a case the optimal schedule for the new extended job sequence can be achieved only by further reducing the completion times of all the jobs. We would like to emphasize here that we are considering the dynamic case at the planning stage when none of the jobs of the original known job sequence has gone to the processing stage.

Let us assume that at any given point of time there are a certain number of jobs ($n$) in a sequence $J$, for which the optimal schedule against a common due-date $D$ on a machine has been already calculated using Algorithm~\ref{main}. In such a case, if there are some additional jobs $n^\prime$ in a sequence $J^\prime$ to be processed against the same due-date and by the same machine without disturbing the sequence $J$, the optimum solution for the new sequence of $n+n^\prime$ jobs in the extended sequence $J+J^\prime$\footnote{$J$ and $J^\prime$ are two disjoint sets of jobs, hence $J+J^\prime$ is the union of two sets maintaining the job sequences in each set.} can be found by further reducing the completion times of jobs in $J$ and the same reduction in the completion times of jobs in $J^\prime$ using Algorithm~\ref{main}. We prove it using Lemma~\ref{dynamic}. 

\begin{lemma}\label{dynamic}
Let, $C_i$ $(i=1,2,\dots,n)$ be the optimal completion times of jobs in sequence $J$ and $C_j^\prime$ $(j=1,2,\dots,n,n+1,\dots,n+n^\prime -1,n+n^\prime)$ be the optimal completion times of jobs in the extended job sequence $J+J^\prime$ with $n+n^\prime$ jobs. Then,
\begin{itemize}
\item[i)] $\exists$ $\gamma \geq 0$ s.t. $C_{i} - C_{i}^\prime = \gamma$ for $i=1,2,\dots,n$
\item[ii)] $ C_{k}^\prime = C_{n} -\gamma + \sum_{\substack{\tau=n+1}}^{k} P_{\tau}$, $(k = n+1,n+2,\dots,n+n^\prime)$\;.
\end{itemize}
\end{lemma}

\begin{proof}
Let $Sol_J$ denote the optimal solution for the job sequence $J$. This optimal value for sequence $J$ is calculated using Algorithm~\ref{main} which is optimal according to Theorem~\ref{theorem}. In the optimal solution let the individual penalties for earliness and tardiness be $E_{i}$ and $T_{i}$, respectively, hence 
\begin{equation}~\label{solj}
Sol_{J}=\sum_{\substack{i=1}}^{n} (\alpha_{i}E_{i}+\beta_{i}T_{i})\;.
\end{equation}
Clearly, the value of $Sol_{J}$ cannot be improved by either reducing the completion times any further as explained in Theorem~\ref{theorem}. Now, processing an additional job sequence $J^\prime$ starting from $C_{n}$ (the completion time of the last job in $J$) means that for the new extended sequence $J+J^\prime$ the tardiness penalty increases further by some value, say $PT_{J'}$. Besides, the due date remains the same, the sequence $J$ is not disturbed and all the jobs in the sequence $J'$ are tardy. Hence the new solution value (say $V_{J+J'}$) for the new sequence $J+J'$ will be

\begin{equation}~\label{vj+j'}
V_{J+J'} = Sol_{J} + PT_{J'}\;.
\end{equation}

For this new sequence we do not need to increase the completion times since it will only increase the tardiness penalty. This can be proved by contradiction. Let $x$ be the increase in the completion times of all the jobs in $J+J^\prime$ with $x>0$. The earliness and tardiness for the jobs in $J$ become $E_i-x$ and $T_i+x$, respectively and the new total penalty ($V_{J}$) for the job sequence $J$ becomes

\begin{equation}
\begin{array}{rcl}
V_{J}& = &\sum\limits_{i=1}^{n} (\alpha_{i}\cdot(E_{i}-x)+\beta_{i}\cdot(T_{i}+x)) \\
     & = &\sum\limits_{i=1}^{n} (\alpha_{i}\cdot E_{i}+\beta_{i}\cdot T_{i}) + \sum\limits_{i=1}^{n} (\beta_{i}-\alpha_{i})\cdot x\;.
\end{array}
\end{equation}

Equation~\eqref{solj} yields

\begin{equation}~\label{vj}
V_{J} = Sol_J + \sum\limits_{i=1}^{n} (\beta_{i}-\alpha_{i})\cdot x \;.
\end{equation}

Since $Sol_{J}$ is optimal $Sol_{J} \leq V_{J}$, we have
\begin{equation}~\label{positive}
\sum\limits_{i=1}^{n} (\beta_{i}-\alpha_{i})\cdot x \geq 0\;.
\end{equation}

Besides, the total tardiness penalty for the sequence $J^\prime$ will further increase by the same quantity, say $\delta$, $\delta \geq 0$. With this shift, the new overall solution value $V_{J+J'}^{\prime}$ will be

\begin{equation}
V_{J+J'}^{\prime} = V_{J} + PT_{J'} + \delta \;.
\end{equation}

Substituting $V_{J}$ from Equation~\eqref{vj} we have
\begin{equation}
V_{J+J'}^{\prime} = Sol_J + \sum\limits_{i=1}^{n} (\beta_{i}-\alpha_{i})\cdot x + PT_{J'} + \delta \;.
\end{equation}

Using Equation~\eqref{vj+j'} gives
\begin{equation}
V_{J+J'}^{\prime} = V_{J+J'} + \sum\limits_{i=1}^{n} (\beta_{i}-\alpha_{i})\cdot x + \delta\;.
\end{equation}

Using Equation~\eqref{positive} and $\delta \geq 0$ we have
\begin{equation}
V_{J+J'}^{\prime}  \geq  V_{J+J'} \;.
\end{equation}

This shows that only a reduction in the completion times of all the jobs can improve the solution. Thus, there exists a $\gamma$, $\gamma \geq 0$ by which the completion times are reduced to achieve the optimal solution for the new job sequence $J+J^\prime$. Clearly, $C_{i}-C_{i}^\prime = \gamma$ for $i=1,2,\dots,n$ and $ C_{k}^\prime = C_{n} -\gamma + \sum_{\substack{\tau=n+1}}^{k} P_{\tau}$, $(k = n+1,n+2,\dots,n+n^\prime)$ since all the jobs are processed one after another without any idle time. 
\end{proof}

\begin{table}
\centering
{\footnotesize 
\caption{Results obtained for the single machine case of the common due date problem and comparison with benchmark results provided in the OR Library~\cite{or}. For any given number of jobs there are $10$ different instances provided and each instance is designated a number $k$. The gray boxes indicate the instances for which our algorithm could not achieve the known solution values given in~\cite{or}.}
\label{result}
\begin{tabular}{ |p{1cm}|p{1cm} p{0.9cm}|p{0.9cm} p{0.9cm}|p{0.9cm} p{0.9cm}|p{0.9cm} p{0.9cm}| }\hline

\textbf{Jobs} & \multicolumn{2}{c|}{\textbf{h=0.2}} & \multicolumn{2}{c|}{\textbf{h=0.4}} & \multicolumn{2}{c|}{\textbf{h=0.6}} & \multicolumn{2}{c|}{\textbf{h=0.8}} \\ \hline \hline
\textbf{n=10} & \textbf{APSA} & \textbf{BR} & \textbf{APSA} & \textbf{BR} & \textbf{APSA} & \textbf{BR} & \textbf{APSA} & \textbf{BR} \\ \hline
\textbf{k=1} & 1936 & 1936 & 1025 & 1025 & 841 & 841 & 818 & 818\\ 
\textbf{k=2} & 1042 & 1042 & 615 & 615 & 615 & 615 & 615 & 615\\ 
\textbf{k=3} & 1586 & 1586 & 917 & 917 & 793 & 793 & 793 & 793\\ 
\textbf{k=4} & 2139 & 2139 & 1230 & 1230 & 815 & 815 & 803 & 803\\ 
\textbf{k=5} & 1187 & 1187 & 630 & 630 & 521 & 521 & 521 & 521\\ 
\textbf{k=6} & 1521 & 1521 & 908 & 908 & 755 & 755 & 755 & 755\\ 
\textbf{k=7} & 2170 & 2170 & 1374 & 1374 & 1101 & 1101 & 1083 & 1083\\ 
\textbf{k=8} & 1720 & 1720 & 1020 & 1020 & 610 & 610 & 540 & 540\\ 
\textbf{k=9} & 1574 & 1574 & 876 & 876 & 582 & 582 & 554 & 554\\ 
\textbf{k=10} & 1869 & 1869 & 1136 & 1136 & 710 & 710 & 671 & 671\\ \hline
\textbf{n=20} & \textbf{APSA} & \textbf{BR} & \textbf{APSA} & \textbf{BR} & \textbf{APSA} & \textbf{BR} & \textbf{APSA} & \textbf{BR} \\ \hline
\textbf{k=1} & 4394 & 4431 & 3066 & 3066 & 2986 & 2986 & 2986 & 2986\\ 
\textbf{k=2} & 8430 & 8567 & 4847 & 4897 & 3206 & 3260 & 2980 & 2980\\ 
\textbf{k=3} & 6210 & 6331 & 3838 & 3883 & 3583 & 3600 & 3583 & 3600\\ 
\textbf{k=4} & 9188 & 9478 & 5118 & 5122 & 3317 & 3336 & 3040 & 3040\\ 
\textbf{k=5} & 4215 & 4340 & 2495 & 2571 & 2173 & 2206 & 2173 & 2206\\ 
\textbf{k=6} & 6527 & 6766 & 3582 & 3601 & 3010 & 3016 & 3010 & 3016\\ 
\textbf{k=7} & 10455 & 11101 & 6279 & 6357 & 4126 & 4175 & 3878 & 3900\\ 
\textbf{k=8} & 3920 & 4203 & 2145 & 2151 & 1638 & 1638 & 1638 & 1638\\ 
\textbf{k=9} & 3465 & 3530 & 2096 & 2097 & 1965 & 1992 & 1965 & 1992\\ 
\textbf{k=10} & 4979 & 5545 & 3012 & 3192 & 2110 & 2116 & 1995 & 1995\\ \hline
\textbf{n=50} & \textbf{APSA} & \textbf{BR} & \textbf{APSA} & \textbf{BR} & \textbf{APSA} & \textbf{BR} & \textbf{APSA} & \textbf{BR} \\ \hline
\textbf{k=1} & 40936 & 42363 & 24146 & 24868 & 17970 & 17990 & 17982 & 17990\\ 
\textbf{k=2} & 31174 & 33637 & 18451 & 19279 & 14217 & 14231 & 14067 & 14132\\ 
\textbf{k=3} & 35552 & 37641 & 20996 & 21353 & 16497 & 16497 & \cellcolor{gray!50}16517 & \cellcolor{gray!50}16497\\ 
\textbf{k=4} & 28037 & 30166 & 17137 & 17495 & 14088 & 14105 & 14101 & 14105\\ 
\textbf{k=5} & 32347 & 32604 & 18049 & 18441 & 14615 & 14650 & 14615 & 14650\\ 
\textbf{k=6} & 35628 & 36920 & 20790 & 21497 & \cellcolor{gray!50}14328 & \cellcolor{gray!50}14251 & 14075 & 14075\\ 
\textbf{k=7} & 43203 & 44277 & 23076 & 23883 & 17715 & 17715 & 17699 & 17715\\ 
\textbf{k=8} & 43961 & 46065 & 25111 & 25402 & 21345 & 21367 & 21351 & 21367\\
\textbf{k=9} & 34600 & 36397 & 20302 & 21929 & 14202 & 14298 & \cellcolor{gray!50}14064 & \cellcolor{gray!50}13952\\ 
\textbf{k=10} & 33643 & 35797 & 19564 & 20048 & 14367 & 14377 & 14374 & 14377\\ \hline
\textbf{n=100} & \textbf{APSA} & \textbf{BR} & \textbf{APSA} & \textbf{BR} & \textbf{APSA} & \textbf{BR} & \textbf{APSA} & \textbf{BR} \\ \hline
\textbf{k=1} & 148316 & 156103 & 89537 & 89588 & 72017 & 72019 & 72017 & 72019\\ 
\textbf{k=2} & 129379 & 132605 & 73828 & 74854 & 59350 & 59351 & 59348 & 59351\\ 
\textbf{k=3} & 136385 & 137463 & 83963 & 85363 & \cellcolor{gray!50}68671 & \cellcolor{gray!50}68537 & \cellcolor{gray!50}68670 & \cellcolor{gray!50}68537\\ 
\textbf{k=4} & 134338 & 137265 & 87255 & 87730 & 69192 & 69231 & 69039 & 69231\\ 
\textbf{k=5} & 129057 & 136761 & 74626 & 76424 & 55291 & 55291 & 55275 & 55277\\ 
\textbf{k=6} & 145927 & 151938 & 81182 & 86724 & 62507 & 62519 & 62410 & 62519\\ 
\textbf{k=7} & 138574 & 141613 & 79482 & 79854 & \cellcolor{gray!50}62302 & \cellcolor{gray!50}62213 & 62208 & 62213\\ 
\textbf{k=8} & 164281 & 168086 & 95197 & 95361 & 80722 & 80844 & 80841 & 80844\\ 
\textbf{k=9} & 121189 & 125153 & 72817 & 73605 & 58769 & 58771 & 58771 & 58771\\ 
\textbf{k=10} & 121425 & 124446 & \cellcolor{gray!50}72741 & \cellcolor{gray!50}72399 & 61416 & 61419 & 61416 & 61419\\ \hline \hline
\end{tabular}
}
\end{table}

\begin{table}
\centering
{\footnotesize
\caption{Results obtained for the single machine case of the common due date problem and comparison with benchmark results provided in the OR Library~\cite{or}. There are $10$ different instances provided and each instance is designated a number $k$.The gray boxes indicate the instances for which our algorithm could not achieve the known solution values given in~\cite{or}.}
\label{result1}
\begin{tabular}{ |p{1.05cm}|p{0.9cm} p{0.9cm}|p{0.9cm} p{0.9cm}|p{0.9cm} p{0.9cm}|p{0.9cm} p{0.9cm}| }\hline

\textbf{Jobs} & \multicolumn{2}{c|}{\textbf{h=0.2}} & \multicolumn{2}{c|}{\textbf{h=0.4}} & \multicolumn{2}{c|}{\textbf{h=0.6}} & \multicolumn{2}{c|}{\textbf{h=0.8}} \\ \hline \hline
\textbf{n=200}  & \textbf{APSA} & \textbf{BR} & \textbf{APSA} & \textbf{BR} & \textbf{APSA} & \textbf{BR} & \textbf{APSA} & \textbf{BR} \\ \hline
\textbf{k=1} & 523042 & 526666 & 300079 & 301449 & 254268 & 254268 & \cellcolor{gray!50}254362 & \cellcolor{gray!50}254268\\ 
\textbf{k=2} & 557884 & 566643 & 333930 & 335714 & \cellcolor{gray!50}266105 & \cellcolor{gray!50}266028 & \cellcolor{gray!50}266549 & \cellcolor{gray!50}266028\\ 
\textbf{k=3} & 510959 & 529919 & 303924 & 308278 & 254647 & 254647 & 254572 & 254647\\ 
\textbf{k=4} & 596719 & 603709 & 359966 & 360852 & \cellcolor{gray!50}297305 & \cellcolor{gray!50}297269 & \cellcolor{gray!50}297729 & \cellcolor{gray!50}297269\\ 
\textbf{k=5} & 543709 & 547953 & 317707 & 322268 & \cellcolor{gray!50}260703 & \cellcolor{gray!50}260455 & 260423 & 260455\\ 
\textbf{k=6} & 500354 & 502276 & 287916 & 292453 & 235947 & 236160 & 236013 & 236160\\ 
\textbf{k=7} & 477734 & 479651 & 279487 & 279576 & 246910 & 247555 & 247521 & 247555\\ 
\textbf{k=8} & 522470 & 530896 & 287932 & 288746 & 225519 & 225572 & \cellcolor{gray!50}225897 & \cellcolor{gray!50}225572\\ 
\textbf{k=9} & 561956 & 575353 & 324475 & 331107 & 254953 & 255029 & 254956 & 255029\\ 
\textbf{k=10} & 560632 & 572866 & 328964 & 332808 & 269172 & 269236 & 269208 & 269236\\ \hline
\end{tabular}
}
\end{table}

\begin{table}
\centering
{\footnotesize
\caption{Average run-times in seconds for the single machine cases for the obtained solutions. The average run-time for any job is the average of all the $40$ instances.}
\label{runtimeS}
\begin{tabular}{ |c|p{0.02cm}cp{0.02cm}|p{0.02cm}cp{0.02cm}|p{0.02cm}cp{0.02cm}|p{0.02cm}cp{0.02cm}|p{0.02cm}cp{0.02cm}|}\hline
\textbf{No. of Jobs} && 10 &&& 20 &&& 50 &&& 100 &&& 200 &\\ \hline
\textbf{BR} && 0.9 &&& 47.8 &&& 87.3 &&& 284.9 &&& 955.2& \\ \hline
\textbf{APSA} && 0.46 &&& 1.12 &&& 22.17 &&& 55.22 &&& 132.32 & \\ \hline
\end{tabular}
}
\end{table}

\section{Results}
In this Section we present our results for the single and parallel machine cases. We used our exact algorithms with simulated annealing for finding the best job sequence. All the algorithms were implemented on MATLAB\textsuperscript{\textregistered} and run on a machine with a $1.73$ GHz processor and $2$ GB RAM. We present our results for the benchmark instances provided by Biskup and Feldmann in~\cite{biskup} for both the single and parallel machine cases. For brevity, we call our approach as APSA and the benchmark results as BR.

We use a modified Simulated Annealing algorithm to generate job sequences and Algorithm~\ref{main} to optimize each sequence to its minimum penalty. Our experiments show that an ensemble size of $4 + n/10$ and the maximum number of iterations as $500\cdot n$, where $n$ is the number of jobs, work best for the provided benchmark instances. The runtime for all the results is the time after which the solutions mentioned in Table~\ref{result} and~\ref{result1} are obtained. The initial temperature is kept as twice the standard deviation of the energy at infinite temperature: $\sigma_{E_{T=\infty}} = \sqrt{\langle E^2 \rangle_{T=\infty} - \langle E \rangle^2_{T=\infty}}$. We estimate this quantity by randomly sampling the configuration space~\cite{salamon}. An exponential schedule for cooling is adopted with a cooling rate of $0.999$. One of the modifications from the standard SA is in the acceptance criterion. We implement two acceptance criteria: the Metropolis acceptance probability, $\min\{1,\exp((-\hspace{-0.3em}\bigtriangleup\hspace{-0.3em} E)/T)\}$~\cite{salamon} and a constant acceptance probability of $0.07$. A solution is accepted with this constant probability if it is rejected by the Metropolis criterion. This concept of a constant probability is useful when the SA is run for many iterations and the metropolis acceptance probability is almost zero, since the temperature would become infinitesimally small. Apart from this, we also incorporate elitism in our modified SA. Elitism has been successfully adopted in evolutionary algorithms for several complex optimization problems~\cite{elitist1,elitist2}. We observed that this concept works well for the CDD problem. As for the perturbation rule, we first randomly select a certain number of jobs in any job sequence and permute them randomly to create a new sequence. The number of jobs selected for this permutation is taken as  $2 + \lfloor\sqrt{n/10}\rfloor$, where $n$ is the number of jobs. For large instances the size of this permutation is quite small but we have observed that it works well with our modified simulated annealing algorithm.

In Table~\ref{result} and~\ref{result1} we present our results (APSA) for the single machine case. The results provided by Biskup and Feldmann can be found in~\cite{feldmann}. The first $40$ instances with $10$ jobs each have been already solved optimally by Biskup and Feldmann and we reach the optimality for all these instances within an average run-time of $0.457$ seconds. 

Among the next $160$ instances we achieve equal results for $13$ instances, better results for $133$ instances and for the remaining $14$ instances with $50$, $100$ and $200$ jobs, our results are within a gap of $0.803$ percent, $0.1955$ percent and $0.1958$ percent respectively. Feldmann and Biskup~\cite{feldmann} solved these instances using three metaheuristic approaches, namely: simulated annealing, evolutionary strategies and threshold accepting; and presented the average run-time for the instances on a Pentium/$90$ PC.
\begin{table}
\centering
{\footnotesize
\caption{Results obtained for parallel machines for the benchmark instances for $k=1$ with $2$, $3$ and $4$ machines up to $200$ jobs.}
\label{result3}
\begin{tabular}{ | c | c | c | c |c|}\hline
\textbf{No. of Jobs} & \textbf{Machines} & \mathversion{bold}$h$\mathversion{normal}\textbf{ value} &  \textbf{Results Obtained} & \textbf{Run-Time (seconds)}  \\ \hline \hline
\multirow{6}{*}{\textbf{10}} & \multirow{2}{*}{2} & 0.4 & 612 & 0.0473 \\ 
& & 0.8 & 398 & 0.0352 \\ \cline{2-5}
& \multirow{2}{*}{3} &  0.4 &  507 & 0.0239 \\ 
& & 0.8 & 256 & 0.0252 \\ \cline{2-5}
& \multirow{2}{*}{4} & 0.4 & 364 & 0.0098 \\ 
& & 0.8 &  197 & 0.0157 \\ \hline
\multirow{6}{*}{\textbf{20}} & \multirow{2}{*}{2} & 0.4 & 1527 & 0.4061 \\ 
& & 0.8 & 1469 & 0.6082 \\ \cline{2-5}
& \multirow{2}{*}{3} &  0.4 &  1085 & 3.4794 \\ 
& & 0.8 & 957 & 7.8108 \\ \cline{2-5}
& \multirow{2}{*}{4} & 0.4 & 848 & 8.5814 \\ 
& & 0.8 &  686 & 8.4581 \\ \hline
\multirow{6}{*}{\textbf{50}} & \multirow{2}{*}{2} & 0.4 & 12911 & 7.780 \\ 
& & 0.8 & 9020 & 55.3845 \\ \cline{2-5}
& \multirow{2}{*}{3} &  0.4 &  8913 & 59.992 \\ 
& & 0.8 & 6010 & 125.867 \\ \cline{2-5}
& \multirow{2}{*}{4} & 0.4 & 7097 & 153.566 \\ 
& & 0.8 &  4551 & 22.347 \\ \hline
\multirow{6}{*}{\textbf{100}} & \multirow{2}{*}{2} & 0.4 & 45451 & 101.475 \\ 
& & 0.8 & 37195 & 147.832 \\ \cline{2-5}
& \multirow{2}{*}{3} &  0.4 &  31133 & 159.872 \\ 
& & 0.8 & 25097 & 186.762 \\ \cline{2-5}
& \multirow{2}{*}{4} & 0.4 & 23904 & 236.132 \\ 
& & 0.8 &  19001 & 392.967 \\ \hline
\multirow{6}{*}{\textbf{200}} & \multirow{2}{*}{2} & 0.4 & 154094 & 165.436 \\ 
& & 0.8 & 133848 & 231.768 \\ \cline{2-5}
& \multirow{2}{*}{3} &  0.4 &  103450 & 226.140 \\ 
& & 0.8 & 96649 & 365.982 \\ \cline{2-5}
& \multirow{2}{*}{4} & 0.4 & 81437 & 438.272 \\ 
& & 0.8 &  71263 & 500.00 \\ \hline \hline
\end{tabular}
}
\end{table}
In Table~\ref{runtimeS} we show our average run-times for the instances and compare them with the heuristic approach considered in~\cite{feldmann}. Apparently our approach is faster and achieves better results. However, there is a difference in the machines used for the implementation of the algorithms. In Table~\ref{result3} we present results for the same problem but with parallel machines for the Biskup benchmark instances. The computation has been carried out for $k=1$ up to $200$ jobs and a different number of machines with restrictive factor $h$. We make a change in the due date as the number of machines increases and assume that the due date $D$ is $D=\lfloor h \cdot \sum_{\substack{i=1}}^{n} P_{i}/m\rfloor$. This assumption makes sense as an increase in the number of machines means that the jobs can be completed much faster and reducing the due-date will test the whole setup for more competitive scenarios. We implemented Algorithm~\ref{parallel} with six different combinations of the number of machines and the restrictive factor. Since these instances have not been solved for the parallel machines, we are presenting the upper bounds achieved for these instances using Algorithm~\ref{parallel} and the modified simulated annealing.

\section{Conclusion and Future Direction}
In this paper we present two novel exact polynomial algorithms for the common due-date problem to optimize any given job sequence. We prove the optimality for the single machine case and the run-time complexity of the algorithms. We implemented our algorithms over the benchmark instances provided by Biskup and Feldmann~\cite{biskup} and the results obtained by using our algorithms are superior to the benchmark results in quality. We discuss how our approach can be used for non-identical parallel machines and present results for the parallel machine case for the same instances. Furthermore, we also discuss the efficiency of our algorithm for a special dynamic case of CDD at the planning stage.

\section*{Acknowledgement}
The research project was promoted and funded by the European Union and the Free State of Saxony, Germany. The authors take the responsibility for the content of this chapter.

\end{document}